\documentclass[journal]{IEEEtran}
\usepackage{pifont}
\usepackage{mathtools}
\usepackage{}
\usepackage{mathrsfs}
\usepackage[noadjust]{cite}
\usepackage{graphicx}
\usepackage{float}
\usepackage{subfigure}
\usepackage{amsmath}
\usepackage{bm}
\usepackage{enumerate}
\usepackage{amssymb}
\usepackage{fixltx2e}
\usepackage{color}
\usepackage{hhline}
\usepackage{algpseudocode}
\newenvironment{proof}{{\indent \indent \it \textbf{Proof}:\quad}}{\hfill $\blacksquare$\par}
\usepackage{graphics}
\usepackage{epsfig}
\usepackage{setspace}
\usepackage{multirow}
\usepackage{ragged2e}
\usepackage{makecell}
\usepackage[ruled]{algorithm2e}

\usepackage{graphics}
\usepackage{epsfig}
\newtheorem{lemma}{\textbf{Lemma}}

\newtheorem{remark}{\textbf{Remark}}

\allowdisplaybreaks[4]
\usepackage{xcolor}
\usepackage{xpatch}

\makeatletter
\def\changeBibColor#1{%
	\in@{#1}{fu2020robust,salem2022secure,liu2015ergodic,chae2014enhanced,wang2022resisting,chiu2013stochastic,shi2024secrecy,kong2018physical,xie2023secrecy,shiu2011physical,goel2008guaranteeing}
	\ifin@\color{red}\else\normalcolor\fi
}
\xpatchcmd\@bibitem
{\item}
{\changeBibColor{#1}\item}
{}{}

\xpatchcmd\@lbibitem
{\item}
{\changeBibColor{#2}\item}
{}{}

\makeatother

\begin{document}

\title{RSMA-Enabled Covert Communications Against Multiple Spatially Random Wardens}
\author{Xinyue~Pei, Jihao~Liu,
	Xuewen~Luo,
	Xingwei~Wang, Yingyang~Chen,~\IEEEmembership{Senior Member,~IEEE,}\\	Miaowen~Wen,~\IEEEmembership{Senior Member,~IEEE,} 	 and Theodoros~A.~Tsiftsis,~\IEEEmembership{Senior Member,~IEEE}
				\thanks{X. Pei, J. Liu, X. Luo, and X. Wang are with School of Computer Science and Engineering, Northeastern University, Shenyang 110819, China (e-mail: peixy@cse.neu.edu.cn,
			2310734@stu.neu.edu.cn, luoxw@cse.neu.edu.cn,
			 wangxw@mail.neu.edu.cn).}
		\thanks{Y. Chen is with the 
			College of Information Science and Technology, Jinan University, Guangzhou 510632, China (e-mail:
			chenyy@jnu.edu.cn). }
			\thanks{M. Wen is with the National Engineering Technology Research Center
				for Mobile Ultrasonic Detection, South China University of Technology,
				Guangzhou 510640, China (e-mail: eemwwen@scut.edu.cn).}
				\thanks{T. A. Tsiftsis is with Department of Informatics Telecommunications,
					University of Thessaly, Lamia 35100, Greece, and also with the Department
					of Electrical and Electronic Engineering, University of Nottingham Ningbo
					China, Ningbo 315100, China (e-mail: tsiftsis@uth.gr).}
}
\markboth{SUBMITTED TO 	IEEE TRANSACTIONS ON VEHICULAR TECHNOLOGY}{}
\maketitle

\begin{abstract} 
This work investigates covert communication in a rate-splitting multiple access (RSMA)-based multi-user multiple-input single-output  system, where the random locations of the wardens  follow a homogeneous Poisson point process. To demonstrate practical deployment scenarios, imperfect channel state information at the transmitter is considered. Closed-form expressions for the statistics of the received signal-to-interference-plus-noise ratio, along with the analytical formulations for the covertness constraint, outage probability, and effective covert throughput (ECT), are derived. Subsequently, an ECT maximization problem is formulated under covertness and power allocation constraints. This optimization problem is addressed using an alternating optimization-assisted genetic algorithm (AO-GA). Simulation results corroborate the theoretical analysis and demonstrate the superiority of RSMA over conventional multiple access schemes, as well as the effectiveness of the proposed AO-GA.
\end{abstract}

\maketitle
\section{Introduction}
Rate-splitting multiple access (RSMA) has emerged as a promising technique for scenarios requiring robust and flexible interference management \cite{dizdar2021rate}. By integrating the rate splitting strategy at the transmitter with successive interference cancellation (SIC) at the receiver, RSMA outperforms space division multiple access (SDMA) and non-orthogonal multiple access (NOMA) in downlink multi-user multiple-input single-output (MU-MISO) scenarios \cite{mao2022rate}.

However, the broadcast nature of wireless channels poses significant threats to the security of RSMA systems. Traditional techniques like cryptographic encryption and physical layer security  have been employed to prevent information from being decoded by wardens \cite{liu2017enhancing}. Nevertheless, these methods cannot prevent 
communication activities from being monitored. In contrast, covert communication (CC) enhances both security and privacy by concealing the existence of transmissions from wardens \cite{ma2021robust}.

Several works have studied the covertness of RSMA systems. Specifically, \cite{hieu2023joint,nguyen2023jamming,chang2024star,zhang2024rate,jia2025robust} investigated optimization strategies for enhancing CC performance in MISO RSMA systems. In parallel,  \cite{zhang2024covert,zhang2025covert,liang2025covert} have conducted performance analyses of RSMA-based CC systems under various architectures, focusing on key metrics such as detection error probability (DEP), outage probability (OP), and covert transmission rate. 

Notably, existing studies \cite{nguyen2023jamming,hieu2023joint,chang2024star,zhang2024rate,zhang2024covert,zhang2025covert,liang2025covert,jia2025robust} primarily consider either a single warden or multiple static wardens. However, in practice, it is challenging for the base station (BS) to obtain accurate information about the number and locations of wardens, especially when they remain silent. Stochastic geometry provides a powerful analytical framework to model the spatial randomness of such unknown and silent wardens \cite{ma2021covert,ma2023covert}. To the best of the authors’ knowledge, CC  performance in RSMA systems with multiple randomly distributed wardens lacks in-depth investigation, which motivates  us in this work.

\begin{figure}[t]
	\centering
	\includegraphics[width=1.5in]{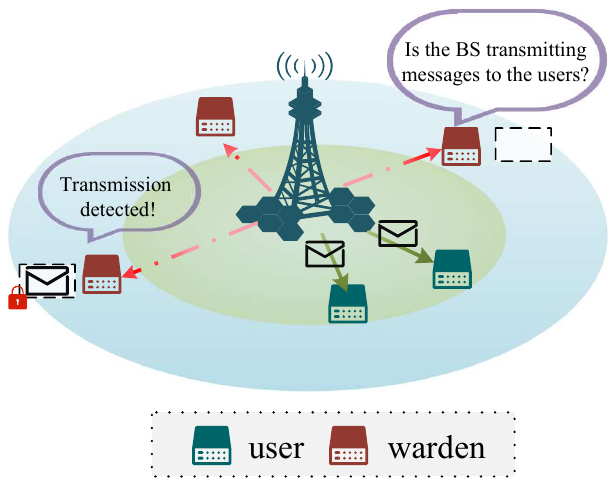}
	\caption{The proposed system model. The green disc region represents the protected zone, and the blue annular region is the area where wardens are distributed.} 	
	\label{System_model}
\end{figure}
In this work, we consider a MU-MISO RSMA-based CC system, where wardens are randomly distributed according to a homogeneous Poisson point process (HPPP). To ensure a  realistic and practical scenario, imperfect channel state information at the transmitter (CSIT) is also taken into account. Our contributions can be summarized as follows:
\begin{itemize}
	\item Closed-form cumulative distribution functions (CDFs) of received signal-to-interference-plus-noise ratios (SINRs) are derived and used to analyze the theoretical expressions for covertness constraint, OP, and effective covert throughput (ECT) under imperfect CSIT and randomly located wardens.
	\item An ECT maximization problem is formulated by jointly optimizing the transmit power and the common power coefficient. Due to the non-convex nature of the objective function and the intractability of its gradients, an alternating optimization-assisted genetic algorithm (AO-GA) is proposed.
	\item Monte Carlo simulations validate the accuracy of the derived theoretical expressions and confirm that the proposed RSMA scheme outperforms NOMA and SDMA. Simulation results also verify the effectiveness of the proposed AO-GA.\footnote{\emph{Notation}: The operations $\Pr(\cdot)$, $|\cdot|$, and $\mathbb{E}\{\cdot\}$ denote the probability, the absolute value,  and the expectation,
		respectively.  $F_X(\cdot)$ and $f_X(\cdot)$ respectively denote the CDF and probability density function (PDF) of a random variable (RV) $X$. $\mathcal{CN}(\mu,\beta)$ denotes a complex Gaussian distribution with mean $\mu$ and variance $\beta$. A gamma distribution with shape parameter $D$ and scale parameter $\theta$ is represented by $\Gamma(D,\theta)$, whose PDF is ${1}/{(\Gamma(D)\theta^D)}x^{D-1}\exp(-{x}/{\theta})$, and CDF is $\gamma(D,x/\theta)/\Gamma(D)$, which can be rewritten as $1-\exp(-x/\theta)\sum_{i=0}^{D-1}1/i!(x/\theta)^i$ with positive integer $D$. $\lfloor \cdot \rceil$ finds the nearest integer. $W_{v,\frac{1}{2}+v}(z)=z^{-v}\exp(z/2)\int_{z}^{\infty}t^{2v}\exp(-t)dt$ is the Whittaker function defined in [\citen{gradshteyn2014table}, (9.224)].}
\end{itemize}

\section{System Model}
We consider a downlink RSMA-based CC system as shown in Fig. \ref{System_model}, where an $M$-antenna BS serves $K$ single-antenna covert users, indexed by $k \in \mathcal{K} = \{1, \ldots, K\}$, with $M \geq K$. Meanwhile, multiple non-colluding passive single-antenna wardens attempt to detect the covert transmission. These wardens, indexed by $v$, are spatially distributed according to a HPPP $\Phi_{\text{E}}$ with density $\lambda_e$, modeling their unknown random locations and numbers.  It is assumed that wardens can be detected near the BS, and thus a protect zone, defined as a disc of radius $r_p$ centered at the BS, is established.\footnote{The BS can create such zone using various detection methods, e.g., metal detectors and X-ray scanners, to identify and eliminate nearby suspicious devices before transmission \cite{liu2017enhancing}.} To enhance covertness, covert users are assumed to be located within this zone. Outside this zone, wardens are spatially distributed over the infinite two-dimensional plane.
 
Without loss of generality, all channels are assumed to experience independent Rayleigh fading with large-scale path loss. The channel from  BS to node $\iota$ ($\iota\in\{k,v\}$) is given by
$\textbf{h}_\iota = \sqrt{\beta_\iota} \textbf{g}_\iota$ ($\textbf{h}_\iota$, $\textbf{g}_\iota\in\mathbb{C}^{M\times 1}$),
where $\beta_{\iota}\triangleq d_{\iota}^{-\alpha}$ models the large-scale fading with $\alpha$ being the path loss exponent, and $\textbf{g}_\iota$ is the small-scale fading vector with independent and identically distributed (i.i.d.) entries modeled as $\mathcal{CN}(0,1)$. Notably, we consider a worst-case secrecy scenario, where the BS has imperfect CSIT of covert users and only statistical knowledge of the wardens' channels, lacking their instantaneous CSI \cite{liu2017enhancing,pei2024secrecy}.  Then, the accurate channel vector from BS to $U_k$ is expressed as  $
 \textbf{g}_k=\sqrt{\epsilon^2}\hat{\textbf{g}}_k+\sqrt{1-\epsilon^2} \triangle\textbf{g}_k$, where $\hat{\textbf{g}}_k$ and $\triangle\textbf{g}_k$ respectively represent the estimated channel and estimation error, each  consisting of  i.i.d. $\mathcal{CN}(0,1)$ entries. The parameter $\epsilon\in[0,1]$ reflects the level of CSIT uncertainty  with a smaller $\epsilon$ indicating worse CSIT quality.

Following RSMA protocol, each message intended for the $k$-th user (denoted by $U_k$) is divided into common and private parts. The common parts of all users are jointly encoded into a common stream $s_{c}$, while the private part of $U_k$ is encoded into private stream $s_k$. It is assumed that $\mathbb{E}[|s_{c}|^2] = \mathbb{E}[|s_k|^2] = 1$, with corresponding power allocation coefficients $a_{c}$ and $a_k$, satisfying $a_{c} + \sum_{k\in\mathcal{K}} a_k = 1$.  Consequently, the signal transmitted by the BS can be derived as $\textbf{s}=\sqrt{ Pa_c}\textbf{w}_cs_c+\sum_{k=1}^{K}\sqrt{ Pa_k}\textbf{w}_ks_k,$ where $P$ represents the total transmit power, while $\textbf{w}_c$, $\textbf{w}_k\in\mathbb{C}^{M\times 1}$ respectively denote the normalized precoders for the common and private streams. 
To eliminate inter-user interference, the private streams utilize zero-forcing beamforming (ZFBF),  ensuring $|\hat{\textbf{g}}_k^H\textbf{w}_j|^2=0$. In contrast, the common stream adopts a random Gaussian beamformer $\textbf{w}_c$ that remains statistically independent of $\hat{\textbf{g}}_k$,  $\triangle\textbf{g}_k$, and $\textbf{w}_k$ \cite{dizdar2021rate,mao2022rate}.\footnote{
Employing ZFBF for the private streams and random precoding for the common stream is proven to be  the optimal Degree of Freedom (DoF)  in MISO broadcast channels scenarios with imperfect CSIT  \cite{mao2022rate}.}

The received signal at $U_k$ can be written as $y_k=\textbf{h}_k^H\textbf{s}+n_k$, where $n_k\sim\mathcal{CN}(0,\sigma^2)$ represents the additive white Gaussian noise (AWGN)  at $U_k$. According to RSMA protocol, $U_k$ first decodes the common stream $s_c$ by treating all private streams as interference. Next, by applying SIC, $U_k$ removes $s_c$ from $y_k$. After successfully canceling $s_c$, $U_k$ decodes its desired private stream $s_k$, treating the private streams of other users as interference. Therefore, the corresponding SINRs for decoding $s_c$ and $s_k$ at $U_k$ are respectively given by
 \begin{small}
	 	\begin{align}
		 		\hspace*{-0.2cm}\gamma_{c,k}\!\!=\!\!\frac{ \mathcal{P}a_c|\textbf{h}_k^H\textbf{w}_c|^2}{\sum_{j\in\mathcal{K}} \mathcal{P}a_j|\textbf{h}_k^H\textbf{w}_j|^2\!+\!1}, \, 
		 		\gamma_{p,k}\!\!=\!\!\frac{ \mathcal{P}a_k|\textbf{h}_k^H\textbf{w}_k|^2}{\sum_{j\in\mathcal{K}\backslash k} \mathcal{P}a_j|\textbf{h}_k^H\textbf{w}_j|^2\!+\!1},
		 	\end{align}
	 \end{small}where $\mathcal{P}=P/\sigma^2$ denotes the transmit signal-to-noise ratio (SNR) at BS.


Then, we focus on the detection process at the wardens. A typical $v$-th warden (denoted by $\text{warden-}v$) monitors the BS’s activity and distinguishes between two hypotheses \cite{ma2021covert}:
\begin{align}\label{hypotheses}
	\begin{cases}
		\mathcal{H}_0\colon & y_{v} = n_{v}, \quad\quad\quad\,\,\,\,\,\, \text{(BS is silent)} \\
		\mathcal{H}_1\colon & y_{v} = \textbf{h}_{v}^H\textbf{s} + n_{v}, \quad \text{(BS is transmitting)}
	\end{cases}			
\end{align}
where $y_{v}$ and $n_v\sim\mathcal{CN}(0,\sigma^2)$ denote  the received signal and the AWGN at $\text{warden-}v$,  respectively.

Based on the above hypotheses, $\text{warden-}v$ performs a binary hypothesis test. Let $\mathcal{D}_0$ and $\mathcal{D}_1$ denote the decisions favoring $\mathcal{H}_0$ and $\mathcal{H}_1$, respectively. The reliability of this test is quantified by the total DEP, defined as \cite{ma2021covert}: $	\xi_{v} = \mathbb{P}_{v,\rm FA} + \mathbb{P}_{v,\rm MD}, $ 
where $\mathbb{P}_{v,\rm FA} = \Pr(\mathcal{D}_1|\mathcal{H}_0)$ and $\mathbb{P}_{v,\rm MD} = \Pr(\mathcal{D}_0|\mathcal{H}_1)$ denote the false alarm and the miss detection probabilities at $\text{warden-}v$, respectively. Under optimal detection, the warden minimizes the DEP to $\xi_v^*$. To guarantee CC, the system must satisfy $\xi_v^* \geq 1 - \varsigma$, where $\varsigma \in [0,1]$ is a predefined covertness tolerance, ensuring that the warden’s detection capability remains below an acceptable threshold.
\begin{table}[]
	\begin{Center}					
		\caption{Shape and Scale Parameters for Different RVs}\vspace{-4mm}  
		\hspace*{-2mm}
		\label{RVs_parameters}
		\scalebox{0.59}{	\begin{tabular}{|c|c|c|}
				\hline\hline
				\multicolumn{1}{|c|}{RV}         & \multicolumn{1}{c|}{Shape parameter} & \multicolumn{1}{c|}{Scale parameter}                 \\ \hline\hline
				$\sum_{j\in\mathcal{K}}a_j|\textbf{g}_k^H\textbf{w}_j|^2$                          & 	$D_{\mathcal{Y},k}=\frac{[(M-K+1)\epsilon^2a_k+(1-a_c)(1-\epsilon^2)]^2}{(M-K+1)\epsilon^4a_k^2+\sum_{j\in\mathcal{K}}a_j^2(1-\epsilon^2)^2}  $                                 & $\hat{\theta}_{\mathcal{Y},k}=\frac{(M-K+1)\epsilon^4a_k^2+\sum_{j\in\mathcal{K}}a_j^2(1-\epsilon^2)^2}{(M-K+1)\epsilon^2a_k+(1-a_c)(1-\epsilon^2)}$                            \\ \hline
				$\sum_{j\in\mathcal{K}}a_j|\textbf{g}_v^H\textbf{w}_j|^2$                          &  $D_1=\frac{(1-a_c)^2}{\sum_{j\in\mathcal{K}}a_j^2}$                                    &   $\hat{\theta}_1=\frac{\sum_{j\in\mathcal{K}}a_j^2}{1-a_c}$                              \\ \hline
				$\sum_{j\in\mathcal{K}\backslash k} a_j|\textbf{g}_k^H\textbf{w}_j|^2$                         &  $	D_{\mathcal{L},k}=\frac{(1-a_c-a_k)^2}{\sum_{j\in\mathcal{K}\backslash k}a_j^2}$                                  &  $\hat{\theta}_{\mathcal{L},k}=\frac{\sum_{j\in\mathcal{K}\backslash k}a_j^2}{1-a_c-a_k}$                               \\ 
				\hline
				$|\textbf{g}_k^H\textbf{w}_k|^2$                           &  $	D_{\mathcal{Z}}=\frac{[(M-K+1)\epsilon^2+1-\epsilon^2]^2}{(M-K+1)\epsilon^4+(1-\epsilon^2)^2}$                                 &  $\hat{\theta}_{\mathcal{Z}}=\frac{(M-K+1)\epsilon^4+(1-\epsilon^2)^2}{(M-K+1)\epsilon^2+1-\epsilon^2}$                               \\ 
				
				\hline\hline
		\end{tabular}}
	\end{Center}
\end{table}
Under the non-colluding detection strategy, each warden independently determines whether the BS is transmitting signals based on its own observation. In this scenario, the CC performance of the system is dominated by the most detrimental warden with the minimal $\xi_v^*$. To ensure covertness, the system must guarantee 
$\xi_w=\min_{v\in\Phi_{\text{E}}}\{ \xi_v^*\}\geq 1 - \varsigma.$

\vspace*{-0.2cm}
\section{Statistical Analysis of SINRs}
To obtain the theoretical expressions of key performance metrics, we first investigate the statistics of received SINRs and obtain the following lemmas.
\begin{lemma}\label{CDF_of_gammack}
The CDF of $\gamma_{c,k}$ can be derived as
\begin{small}
\begin{align}\label{closed_form_cdf_of_gammack}
	F_{\gamma_{c,k}}(x)=1-\exp(-\varphi_kx)(\mu_kx+1)^{-D_{\mathcal{Y},k}},
\end{align}
\end{small}where $\mu_k={\hat{\theta}_{\mathcal{Y},k}}/{a_c}$, $	\varphi_k=1/{( \mathcal{P}a_c\beta_k)}$; $D_{\mathcal{Y},k}$ and $\hat{\theta}_{\mathcal{Y},k}$ have been defined in Table \ref{RVs_parameters}.
\end{lemma}
\begin{proof}
	See Appendix \ref{Proof_of_cdf_gammac}.
\end{proof}
\begin{lemma}\label{CDF_of_gammapk}
	The CDF of $\gamma_{p,k}$ can be derived as
\begin{small}
	\begin{align}\label{closed_form_cdf_of_gammapk}
		F_{\gamma_{p,k}}(x) = 1 - \sum_{m=0}^{\lfloor D_{\mathcal{Z}} \rceil - 1} \sum_{n=0}^{\lfloor D_{\mathcal{Z}} \rceil - m - 1}
		\frac{ \exp\left(-\frac{x}{\theta_{\mathcal{Z}}}\right) x^{m+n}\psi_k}
		{\left( x\theta_{\mathcal{L},k} + \theta_{\mathcal{Z}} \right)^{n + D_{\mathcal{L},k}}  },
	\end{align}
\end{small}where
$\psi_k=\frac{\theta_{\mathcal{L},k}^n\theta_{\mathcal{Z}}^{D_{\mathcal{L},k}-m}\Gamma(n+D_{\mathcal{L},k})}{m!n!\Gamma(D_{\mathcal{L},k})}$,  $\theta_{\mathcal{L},k}=\mathcal{P}\beta_k(1-\epsilon^2)\hat{\theta}_{\mathcal{L},k}$, and ${\theta}_{\mathcal{Z}}=\mathcal{P}a_k\beta_k\hat{\theta}_{\mathcal{Z}}$; 
$D_{\mathcal{L},k}$, $D_{\mathcal{Z}}$, $\hat{\theta}_{\mathcal{L},k}$, and $\hat{\theta}_{\mathcal{Z}}$ have been defined in Table \ref{RVs_parameters}.
\end{lemma}
\begin{proof}
	See Appendix \ref{Proof_of_cdf_gammapk}.
\end{proof}
\vspace*{-0.2cm}
\section{Performance Analysis}
In this section, we investigate the key metrics for CC performance of RSMA systems, including covertness constraint, OP, and the ECT.
\vspace*{-0.4cm}
\subsection{Covertness Constraint}
In this subsection, we aim to analyze the system’s covertness by deriving the closed-form expression for the minimum DEP across all wardens, i.e., $\xi_w=\min_{v\in\Phi_{\text{E}}}\{ \xi_v^*\}\geq 1 - \varsigma$. However,
deriving  $\xi_v^*$
is analytically intractable, which complicates the analysis of $\xi_w$. To solve this challenge,  we adopt Pinsker’s inequality, yielding a lower bound for $\xi_v^*$ \cite{ma2021covert,ma2021robust,ma2023covert,zhang2024rate}:
\begin{small}
\begin{align}\label{KL_lower_bound}
	\xi_v^*\geq 1-\sqrt{\frac{1}{2}\mathcal{D}(p_1(y_v)\|p_0(y_v))},
\end{align}
\end{small}where $p_0(y_v)$ and $p_1(y_v)$ respectively denote the likelihood functions of the received signals of $\text{warden-}v$ under $\mathcal{H}_0$
and $\mathcal{H}_1$, and $\mathcal{D}(p_1\|p_0)$ denotes the Kullback-Leibler (KL) divergence  from $p_1(y_v)$ to
$p_0(y_v)$. Then substituting (\ref{KL_lower_bound}) into $\xi_w\geq 1 - \varsigma$, the system covertness can be guaranteed by satisfying the equivalent condition: $\max_{v\in\Phi_{\text{E}}}\left\{\sqrt{\frac{1}{2}\mathcal{D}(p_1(y_v)||p_0(y_v))}\right\}\leq \varsigma$.

Recall that wardens are passive, and thus their instantaneous CSIs are unknown to the BS. Recall that only the statistical CSIs of the wardens are available. To this end, the average KL divergence is adopted in this work, where the expectation is taken over the spatial distribution and the small-scale fading statistics of wardens. Accordingly, the CC constraint for $\text{warden-}v$ can be expressed as \cite{ma2021covert,ma2023covert}:
\begin{small}
\begin{align}\label{average_covert_constraint}
\mathbb{E}\left\{\max_{v\in\Phi_{\text{E}}}\left\{\sqrt{\frac{1}{2}\mathcal{D}(p_1(y_v)||p_0(y_v))}\right\}\right\}\leq \varsigma.
\end{align}
\end{small}Then we can achieve a tractable expression of covertness constraint in the following lemma:
\begin{lemma}\label{expression_of_covertenss_constraint}
The expressions for covertness constraint considering HPPP distributed and non-colluding wardens can be derived as $\nu\mathcal{P}\leq \varsigma$,
where 
\begin{small}
\begin{align}
\nu=\frac{(\lambda_e\pi)^{\frac{\alpha}{4}}\exp\left(\frac{\lambda_e\pi r_p^2}{2}\right)W_{-\frac{\alpha}{4},\frac{1}{2}-\frac{\alpha}{4}}\left(\lambda_e\pi r_p^2\right)}{2r_p^{\frac{\alpha}{2}}}.
\end{align}
\end{small}
\end{lemma}
\begin{proof}
	See Appendix \ref{proof_of_covertness}.
\end{proof}
\begin{remark}
	It is observed that $\nu$ depends only on path loss exponent $\alpha$, protected zone radius $r_p$, and warden density $\lambda_e$, regardless of the power allocation  coefficients or the number of transmit antennas.
\end{remark}

\vspace*{-0.5cm}
\subsection{Outage Probability Analysis}
Given that OP is crucial, we present its analysis below. Under RSMA protocol, an outage occurs at $U_k$ when it fails to decode either $s_c$ or $s_k$. Thus the OP of $U_k$ can be expressed as
\begin{small}
\begin{align}
P_{\text{out},k}
=&\Pr\left(\gamma_{c,k}<\gamma_{k,th}^{s_c}\right)+\Pr\left(\gamma_{p,k}<\gamma_{k,th}^{s_k}\right)\nonumber\\
&-\underbrace{\Pr\left(\gamma_{c,k}<\gamma_{k,th}^{s_c},\gamma_{p,k}<\gamma_{k,th}^{s_k}\right)}_{I_p},
\end{align}
\end{small}where $\gamma_{k,th}^{s_c}$ and $\gamma_{k,th}^{s_p}$ are the SINR thresholds of $s_c$ and $s_k$ at $U_k$,  respectively. Specifically, the theoretical expression of $P_{\text{out},k}$ is summarized in the following lemma.
\begin{lemma}
	The OP of $U_k$ can be derived as
\begin{small}
	\begin{align}\label{Pout}
	P_{\text{out},k}=\max\left\{	F_{\gamma_{c,k}}(\gamma_{k,th}^{s_c}),	F_{\gamma_{p,k}}(\gamma_{k,th}^{s_k})\right\}.
	\end{align}
\end{small} 
\end{lemma}
\begin{proof}
	Due to the shared structure and the presence of Gamma-distributed terms, characterizing $I_p$ becomes intractable. 
Following methods in \cite{zhang2024covert,zhang2025covert,liang2025covert}, $I_p$ can be tightly approximated as $\min\{F_{\gamma_{c,k}}(\gamma_{k,th}^{s_c}),F_{\gamma_{p,k}}(\gamma_{k,th}^{s_k})\}$, which enables a tractable expression of $P_{\text{out},k}$.
\end{proof}
\begin{figure*}[htbp] 
	\centering
	\begin{minipage}[t]{0.24\linewidth} 
		\centering
		\includegraphics[width=0.92\linewidth]{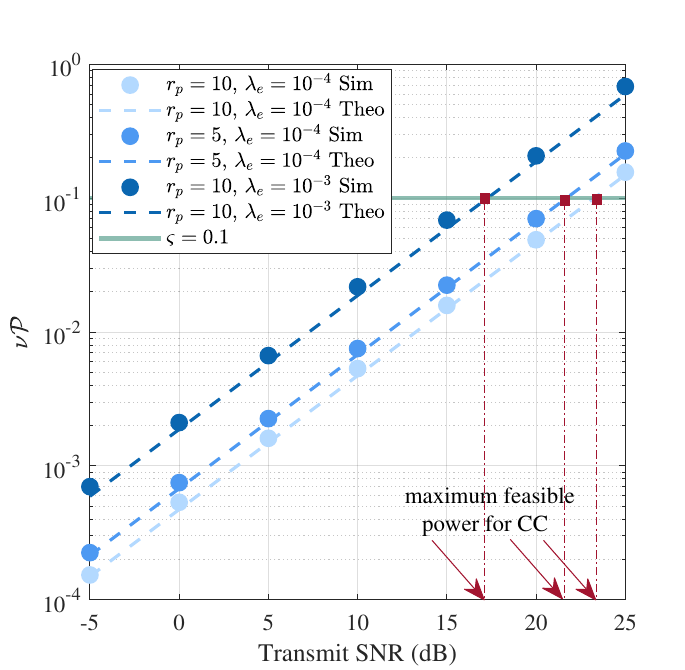}
		\caption{Covertness constraint versus transmit SNR $\mathcal{P}$ with different $r_p$ and $\lambda_e$.}
		\label{CC_vs_SNR}
	\end{minipage}%
	\hfill
	\begin{minipage}[t]{0.24\linewidth}
		\centering
		\includegraphics[width=0.935\linewidth]{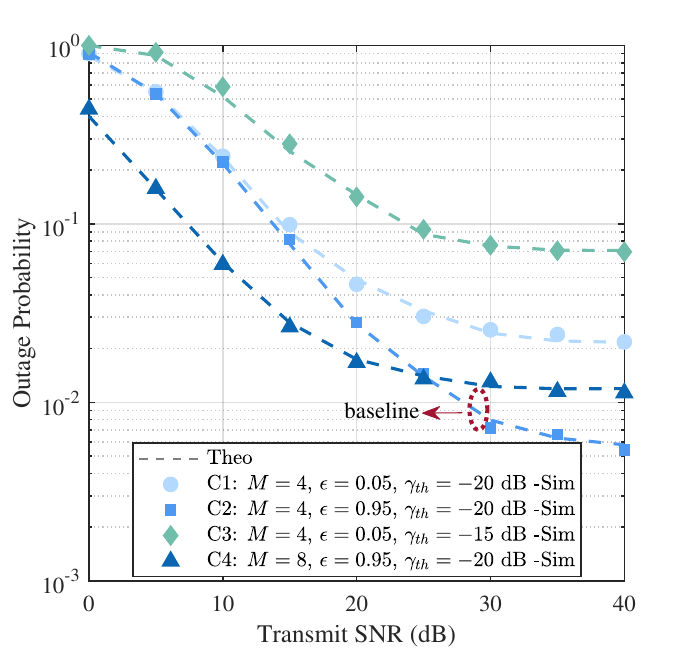}
		\caption{OP versus transmit SNR $\mathcal{P}$ in different Cases (denoted by C1-C4).}
		\label{outage_SNR}
	\end{minipage}%
	\hfill
	\begin{minipage}[t]{0.24\linewidth}
		\centering
		\includegraphics[width=0.93\linewidth]{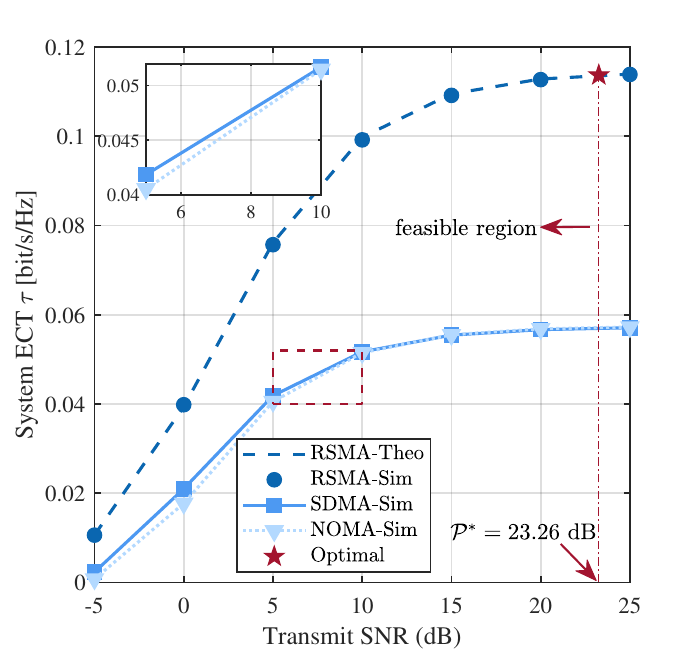}
		\caption{System ECT versus transmit SNR $\mathcal{P}$.}
		\label{ECT_SNR}
	\end{minipage}%
	\hfill
	\begin{minipage}[t]{0.24\linewidth}
		\centering
		\includegraphics[width=0.95\linewidth]{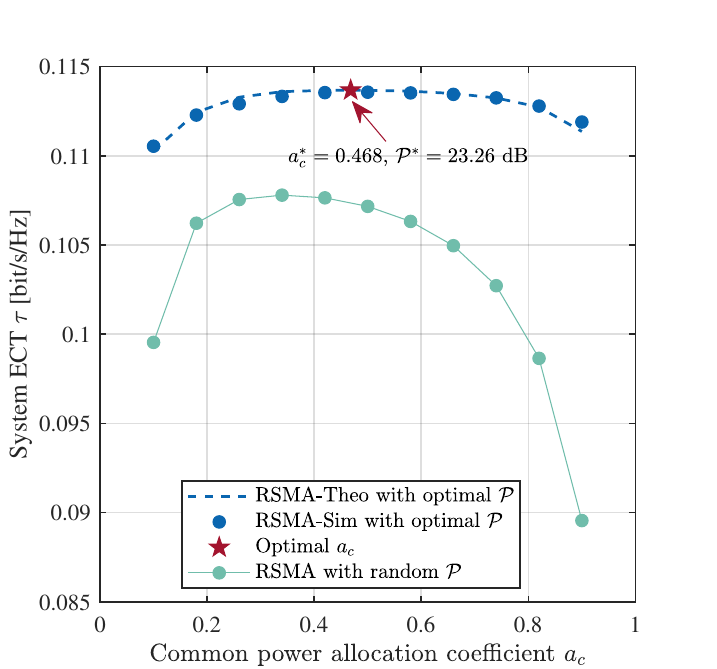} 
		\caption{System ECT versus common power allocation coefficient $a_c$.}
		\label{ECT_vs_ac}
	\end{minipage}
\end{figure*}
\vspace*{-0.3cm}
\subsection{Effective Covert Throughput Analysis}
This subsection focuses on maximizing the ECT, which is defined as the achievable effective throughput under the covertness constraint \cite{ma2021covert,ma2023covert}. Based on the decoding order of $s_c$ and $s_k$, the effective throughput of $U_k$ is given by \cite{mao2022rate,vu2023performance}:
 \begin{small}
 	\begin{align}
 		\tau_k = R_{c,k} \Pr\left(\gamma_{c,k} \geq \gamma_{k,th}^{s_c}\right) + R_{p,k} \left(1 - P_{{\rm out},k}\right),
 	\end{align}
 \end{small}where $R_{c,k} = \log_2(1 + \gamma_{k,th}^{s_c})$ and $R_{p,k} = \log_2(1 + \gamma_{k,th}^{s_k})$ denote the target rates of $s_c$ and $s_k$, respectively. To maximize ECT, we optimize the transmit SNR $\mathcal{P}$ and the power allocation coefficient $a_c$. 
  For tractability, we assume $a_k={(1-a_c)}/{K}$, which is reasonable for symmetric users with identical channel statistics \cite{dizdar2021rate}. Then the optimization problem can be formulated as
\begin{small}
	\begin{align}
	\textbf{P}_0:\; \max_{\mathcal{P},\,a_c} \tau=\sum_{k=1}^{K}\tau_k \; \text{s.t.}\; \nu\mathcal{P} \leq \varsigma, a_c\in[0,1), \mathcal{P}>0.
\end{align}
\end{small}However, from Eqs. (\ref{closed_form_cdf_of_gammack}), (\ref{closed_form_cdf_of_gammapk}), and (\ref{Pout}), it is evident that $\tau$ is a highly complex and non-convex function with respect to both $\mathcal{P}$ and $a_c$. The lack of analytical gradients and unclear monotonicity further hinder the use of conventional convex optimization or gradient-based methods. To address this, we first adopt the AO approach to decompose  $\textbf{P}_0$ into two simpler sub-problems:
\begin{small}
\begin{align}
\textbf{P}_{1}\!:\! \max_{\mathcal{P}} \tau \text{ s.t. } \mathcal{P} \in (0, \varsigma/\nu], \quad 
\textbf{P}_{2}\!:\! \max_{a_c} \tau \text{ s.t. } a_c \in [0, 1).
\end{align}
\end{small}Unfortunately, both $\textbf{P}_1$ and $\textbf{P}_2$ inherit the non-convexity and gradient intractability of $\textbf{P}_0$. Therefore, we resort to the GA to efficiently solve these subproblems. The overall optimization procedure is summarized in Algorithm \ref{AO_with_GA}. 
\begin{remark}
The overall computational complexity of Algorithm \ref{AO_with_GA} is $\mathcal{O}(I_{\text{AO}}(\sum_{i=1}^{2}N_{g,i}(N_{p,i}-N_{e,i})\lfloor D_{\mathcal{Z}} \rceil^2K))$, where $I_{\text{AO}}$ is the number of AO iterations, and the GA complexity for each $\textbf{P}_i$ depends on population size $N_{p,i}$, elitism size $N_{e,i}$, generations $N_{g,i}$, and fitness function $-\tau$.
\end{remark}
\begin{algorithm}[t]
	\begin{small}	
		\caption{AO-GA} 
		\label{AO_with_GA} 
		\LinesNumbered 
		\KwIn{System parameters; tolerance $e$; max iterations $\mathcal{N}_{\max}$.} 
		\KwOut{Optimized $a_c^*$ and $\mathcal{P}^*$.} 
		Randomly initialize $a_c^{(0)}\in[0,1)$ and $\mathcal{P}^{(0)}\in(0,\varsigma/\nu]$, set iteration index $n=0$\;		
		\Repeat{\text{convergence or} $n>\mathcal{N}_{\max}$}{
			
			Fix $a_c^{(n)}$, solve  $\textbf{P}_{1}$ using GA, obtain 
			$\mathcal{P}^{(n+1)}$\;
			
			Fix $\mathcal{P}^{(n+1)}$, solve $\textbf{P}_{2}$ using GA, obtain
			$a_c^{(n+1)}$\;

			\If{$|\tau^{(n+1)} - \tau^{(n)}| < e$}{
				\textbf{break}\;
			}
			$n \leftarrow n + 1$\;
		}

		$a_c^* \leftarrow a_c^{(n)}$, $\mathcal{P}^* \leftarrow \mathcal{P}^{(n)}$\;
	\end{small}
\end{algorithm}

\section{Numerical Results}
In this section, we verify the accuracy of theoretical expressions through Monte Carlo methods. Unless otherwise stated, the simulation parameters are set as follows \cite{dizdar2021rate,zhang2025covert,ma2021covert,ma2023covert,pei2024secrecy}: $M=4$, $K=4$;  $\epsilon=0.95$; $a_c=0.5$ and $a_k=(1-a_c)/K$; $\alpha=2$, $d_k=5$ m, $r_p$=10 m; $\varsigma=0.1$; $\gamma_{k,th}^{s_c}=\gamma_{k,th}^{s_k}=\gamma_{th}=-20$ dB; $e=10^{-9}$ and $\mathcal{N}_{\max}=10$.

In Fig. \ref{CC_vs_SNR}, we plot the covertness constraint versus transmit SNR for different values of the radius of protected zone $r_p$ and the density of warden distribution $\lambda_e$.
It can be observed that $\nu\mathcal{P}$ is a monotonically increasing function of $\mathcal{P}$. As such, increasing $\lambda_e$ enhances the multi-user diversity gain when selecting the most detrimental warden, while decreasing $r_p$ strengthens the eavesdropping channel. Both operations lead to an increase in $\nu$, thereby reducing the maximum feasible transmit power that satisfies the covertness constraint.

In Fig. \ref{outage_SNR}, we compare the outage performance under four cases: C1 increases CSIT uncertainty; C2 adopts the baseline settings; C3 further raises the SINR threshold on top of C1; C4 increases the number of transmit antennas. Results show that higher channel uncertainty and SINR thresholds degrade outage performance. Notably,
the 8-antenna RSMA system only outperforms its 4-antenna
counterpart at low-to-medium SNR. Despite offering more
spatial DoF, the 8-antenna scheme is more sensitive to CSIT
uncertainty. At high SNR, even small estimation errors can
cause beam misalignment and offset the benefits of high-dimensional precoding.

Fig. \ref{ECT_SNR} shows the system ECT versus $\mathcal{P}$ with the optimal $a_c^* = 0.468$ obtained via AO-GA. We compare the performance of the proposed RSMA scheme with that of SDMA using ZFBF (by setting $a_c = 0$), and NOMA with a random precoder, where $a_k$ is determined through exhaustive search.
Clearly, RSMA outperforms both SDMA and NOMA, owing to its robustness and its ability to achieve the optimal DoF. In addition, it is observed that NOMA slightly performs worse than SDMA, as it can suffer from a loss of sum DoF due to the inefficient utilization of SIC when $M>=K$.
Furthermore, the effectiveness of AO-GA is verified in Fig.~\ref{ECT_vs_ac}, where the optimal point is found at $\mathcal{P}^* = 23.26$ dB and $a_c^* = 0.468$.

\section{Conclusions}
In this work, we have investigated a RSMA-based MU-MISO system consisting of a multi-antenna BS, several users, and multiple non-colluding HPPP-distributed wardens.
Considering imperfect CSIT, we have derived theoretical expressions for the covertness constraint, OP, and ECT. Moreover, we have proposed the AO-GA to solve the system ECT optimization problem by jointly optimizing the transmit SNR and the common power allocation coefficient.
The theoretical findings and the effectiveness of the AO-GA have been validated through simulation results.
We have also explored the impacts of transmit SNR, HPPP density, CSIT uncertainty, protected zone radius, and SINR thresholds. 
Finally, simulation results have demonstrated that the proposed RSMA scheme can effectively mitigate the impact of CSIT uncertainty and exhibits superior covert performance against spatially random wardens compared to SDMA and NOMA.

\begin{appendices}
	\setcounter{equation}{0}
	\renewcommand{\theequation}{\thesection.\arabic{equation}}

	\section{ Proof of Lemma \ref{CDF_of_gammack}}\label{Proof_of_cdf_gammac}
		\setcounter{equation}{0}
 For convenience,  define $\mathcal{X}=\mathcal{P}a_c\beta_k|\textbf{g}_k^H\textbf{w}_c|^2$ and $\mathcal{Y}=\sum_{j\in\mathcal{K}} \mathcal{P}a_j\beta_k|\textbf{g}_k^H\textbf{w}_j|^2$.  Since $\textbf{w}_c$ is isotropically distributed and independent of $\textbf{g}_k$,  $\mathcal{X} \sim \Gamma(1, \mathcal{P}a_c\beta_k)$ \cite{dizdar2021rate}. To characterize the distribution of $\mathcal{Y}$, we first decompose the imperfect channel as $
\textbf{g}_k=\sqrt{\epsilon^2}\hat{\textbf{g}}_k+\sqrt{1-\epsilon^2} \triangle\textbf{g}_k$. For tractability, we ignore the cross-terms involving both  $\hat{\textbf{g}}_k$ and $\triangle\textbf{g}_k$. Therefore, we have the following approximations:
\begin{small}
\begin{subequations}
	\begin{align}	a_k\beta_k|\textbf{g}_k^H\textbf{w}_k|^2\approx&a_k\beta_k\left(\epsilon^2|\hat{\textbf{g}}_k^H\textbf{w}_k|^2+(1\!-\!\epsilon^2)|\triangle\textbf{g}_k^H\textbf{w}_k|^2\right), \label{app_of_hkwk}\\
		a_j\beta_k|\textbf{g}_k^H\textbf{w}_j|^2\approx&a_j\beta_k(1-\epsilon^2)|\triangle\textbf{g}_k^H\textbf{w}_j|^2.
	\end{align}
\end{subequations} 
\end{small}Since $\textbf{w}_k$ is isotropic and independent of $\triangle\textbf{g}_k$, we have $a_j(1-\epsilon^2)|\triangle\textbf{g}_k^H\textbf{w}_j|^2\sim\Gamma(1,a_j(1-\epsilon^2))$ and $a_k(1-\epsilon^2)|\triangle\textbf{g}_k^H\textbf{w}_k|^2\sim\Gamma(1,a_k(1-\epsilon^2))$. Furthermore,
$a_k\epsilon^2|\hat{\textbf{g}}_k^H\textbf{w}_k|^2\sim\Gamma(M-K+1,a_k\epsilon^2)$ \cite{jaramillo2014coordinated}. 

By using moment-matching method in \cite{jaramillo2014coordinated}, the sum $\sum_{j\in\mathcal{K}} \Gamma(D_j,\theta_j)$ is approximated by $\Gamma(D,\theta)$, where $D = \frac{(\sum_j D_j\theta_j)^2}{\sum_j D_j\theta_j^2}$ and $\theta = \frac{\sum_j D_j\theta_j^2}{\sum_j D_j\theta_j}$. Accordingly, we approximate $\mathcal{Y} \sim \Gamma(D_{\mathcal{Y},k}, \hat{\theta}_{\mathcal{Y},k}\mathcal{P}\beta_k)$, and then derive the CDF of $\gamma_{c,k}$ as 
\begin{small}
\begin{align}\label{A2}
\hspace*{-0.3cm}	F_{\gamma_{c,k}}(x)= 1 - \frac{\exp\left(-\frac{x}{\mathcal{P}a_c\beta_k}\right)}
	{\theta_{\mathcal{Y},k}^{D_{\mathcal{Y},k}}\Gamma(D_{\mathcal{Y},k})}
	\int_{0}^{\infty} \exp(-I_0 y) \, y^{D_{\mathcal{Y},k}-1} dy,
\end{align}
\end{small}where $I_0 = \frac{x}{\mathcal{P}a_c\beta_k} + \frac{1}{\theta_{\mathcal{Y},k}}$. Using [\citenum{gradshteyn2014table}, Eq. (3.351.3)], (\ref{A2}) can be rewritten as (\ref{closed_form_cdf_of_gammack}), completing the proof.

\vspace*{-0.2cm}
	\section{Proof of Lemma \ref{CDF_of_gammapk}}\label{Proof_of_cdf_gammapk}
Let $\mathcal{L}=\sum_{\mathcal{K}\backslash k} \mathcal{P}a_j\beta_k|\textbf{g}_k^H\textbf{w}_j|^2$ and $\mathcal{Z}=\mathcal{P}a_k\beta_k|\textbf{g}_k^H\textbf{w}_k|^2$.
Following the moment-matching method in Appendix~\ref{Proof_of_cdf_gammac}, the distributions of $\mathcal{L}$ and $\mathcal{Z}$ can be approximated by $\Gamma(D_{\mathcal{L},k},\theta_{\mathcal{L},k})$ and $\Gamma(D_{\mathcal{Z}},\theta_{\mathcal{Z}})$, respectively. Accordingly, the CDF of $\gamma_{p,k}$ is given by
\begin{small}
	\begin{align}\label{initial_form_of_Fgammapk}
		F_{\gamma_{p,k}}(x)
	=\int_{0}^{\infty}\underbrace{\frac{\gamma\left(D_{\mathcal{Z}},\frac{x(l+1)}{\theta_{\mathcal{Z}}}\right)}{\Gamma(D_{\mathcal{Z}})}}_{I_1}
		\frac{l^{D_{\mathcal{L}}-1}\exp\left(-\frac{l}{\theta_{\mathcal{L}}}\right)}{\Gamma(D_{\mathcal{L}})\theta_{\mathcal{L}}^{D_{\mathcal{L}}}}\,dl.
	\end{align}
\end{small}

Since the closed-form expression for the above integral is not available, we approximate $D_{\mathcal{Z}}$ by $\lfloor D_{\mathcal{Z}} \rceil$. Then $I_1$ can be rewritten as
\begin{small}
	\begin{align}\label{after_form_of_I1}
		I_1 
		&\overset{(a)}{=} 1 - \exp\left(-\frac{x(l+1)}{\theta_{\mathcal{Z}}}\right)
		\sum_{m=0}^{\lfloor D_{\mathcal{Z}} \rceil - 1}
		\sum_{k=m}^{\lfloor D_{\mathcal{Z}} \rceil - 1}
		\binom{k}{m} \frac{1}{k!} \left(\frac{x}{\theta_{\mathcal{Z}}}\right)^k l^{k-m} \nonumber \\
		&\overset{(b)}{=} 1 - \sum_{m=0}^{\lfloor D_{\mathcal{Z}} \rceil - 1}
		\sum_{n=0}^{\lfloor D_{\mathcal{Z}} \rceil - m - 1}
		\kappa\, \exp\left(-\frac{x}{\theta_{\mathcal{Z}}}l\right) l^n,
	\end{align}
\end{small}where $(a)$ is derived based on the CDF of $\Gamma(\lfloor D_{\mathcal{Z}} \rceil,\theta_{\mathcal{Z}})$ and the binomial expansion, and $(b)$ stands for $n = k - m$ and $\kappa = \frac{\exp\left(-\frac{x}{\theta_{\mathcal{Z}}}\right)\left(\frac{x}{\theta_{\mathcal{Z}}}\right)^{m+n}}{m!n!}$. Substituting  (\ref{after_form_of_I1}) into (\ref{initial_form_of_Fgammapk}) and applying~\cite[Eq. (3.326.2)]{gradshteyn2014table}, we finally arrive at (\ref{closed_form_cdf_of_gammapk}).

\vspace*{-0.1cm}
\section{Proof of Lemma \ref{expression_of_covertenss_constraint}}\label{proof_of_covertness}
\setcounter{equation}{0}
Before deriving the closed-form expression for the covertness constraint, we first characterize the probability distribution of $y_v$. According to (\ref{hypotheses}), the likelihood functions of $y_v$ under hypotheses $\mathcal{H}_0$ and $\mathcal{H}_1$ are given by
$p_0(y_v)={\exp(-|y_v|^2/\lambda_0)}/{(\pi\lambda_0)}$ and $p_1(y_v)={\exp(-|y_v|^2/\lambda_0)}/{(\pi\lambda_1)}$, where
$\lambda_0 = \sigma^2$ and $\lambda_1 = Pa_c|\textbf{h}_v^H\textbf{w}_c|^2 + \sum_{k=1}^{K} Pa_k|\textbf{h}_v^H\textbf{w}_k|^2+\sigma^2$. Then  the KL divergence  is given by [\citen{ma2021robust}, (11b)]:
\begin{small}
\begin{align}\label{KL_1}
	\hspace*{-0.26cm}\mathcal{D}(p_1(y_v)||p_0(y_v))\!\!=\!\!\!\int_{-\infty}^{+\infty}\!\!p_1(y_v)\ln\!\frac{p_1(y_v)}{p_0(y_v)}\!dy_v\!\!=\!\!m\!\!-\!\!\ln\!\left(1\!+\!m\right),
\end{align}
\end{small}where $m=\frac{\lambda_1-\lambda_0}{\lambda_0}=\mathcal{P}d_v^{-\alpha}\mathcal{M}$ and  $\mathcal{M}=a_c|\textbf{g}_v^H\textbf{w}_c|^2+\sum_{j\in\mathcal{K}}a_j|\textbf{g}_v^H\textbf{w}_j|^2$. Since $\textbf{w}_c$ and $\textbf{w}_j$ are isotropically  distributed and independent of $\textbf{g}_v$, RVs $a_c|\textbf{g}_v^H\textbf{w}_c|^2$ and $\sum_{j\in\mathcal{K}}a_j|\textbf{g}_v^H\textbf{w}_j|^2$ respectively follow  $\Gamma(1,a_c)$ and $\Gamma(D_1,\hat{\theta}_1)$. As a result,  $\mathcal{M}$ can be approximated by $\Gamma(D_{\mathcal{M}},\theta_{\mathcal{M}})$, where $D_{\mathcal{M}} = \frac{(a_c+D_1\hat{\theta}_1)^2}{a_c^2+D_1\hat{\theta}_1^2}$ and $\theta_{\mathcal{M}} = \frac{a_c^2+D_1\hat{\theta}_1^2}{a_c+D_1\hat{\theta}_1}$.

Using the inequality $\ln(1 + x) \geq x - {x^2}/{2}$ for $x > 0$, we obtain an upper bound for (\ref{KL_1}), i.e., $\mathcal{D}(p_1(y_v)||p_0(y_v))\leq m^2/2$. Substituting this bound into (\ref{average_covert_constraint}), we arrive at $\zeta=\max_{v\in\Phi_{\text{E}}}\mathbb{E}_{\Phi_{\text{E}}}\left\{d_v^{-\alpha}\right\}\mathcal{P}m_0/2\leq\varsigma$, where $m_0=D_{\mathcal{M}}\theta_{\mathcal{M}}=1$ is the mean of RV $\mathcal{M}$. In the following, we show how to evaluate $\zeta$. 
\begin{small}
\begin{align}
\zeta&=\frac{\mathcal{P}m_0}{2}\int_{r_p}^{\infty}r^{-\alpha}f_{d_v}(r)dr\nonumber\\
&\overset{(a)}{=}\frac{\mathcal{P}m_0\lambda_e\pi}{\exp(-\lambda_e\pi r_p^2)}\int_{r_p}^{\infty}r^{1-\alpha}\exp(-\lambda_e\pi r^2)dr\nonumber\\
&\overset{(b)}{=}\frac{\mathcal{P}m_0(\lambda_e\pi)^{\alpha/2}}{2\exp(-\lambda_e\pi r_p^2)}\int_{\pi\lambda_er_p^2}^{\infty}\frac{\exp(-z)}{z^{\frac{\alpha}{2}}}dz,
\end{align}
\end{small}where $(a)$ can be arrived by using  void probability over a HPPP \cite{haenggi2013stochastic}; $(b)$ comes from $z=\lambda_e\pi r_p^2$. Finally, by applying [\citen{gradshteyn2014table}, Eq. (3.381.6)], we arrive at $\nu\mathcal{P}\leq \varsigma$. 
\end{appendices}

\def\bibfont{\fontsize{7.5}{8}\selectfont}
\bibliographystyle{IEEEtran}
\bibliography{refer.bib}
\end{document}